\documentclass[10pt]{llncs}

\usepackage{graphicx}
\usepackage{url}
\usepackage{amsmath}
\usepackage{graphicx}
\usepackage{subfig}
\usepackage{hyperref}
\usepackage{dsfont}
\usepackage{color}
\usepackage{tikz}
\usepackage{algorithm}
\usepackage{algorithmicx}
\usepackage[noend]{algpseudocode}
\usepackage{mathtools}
\usepackage{multirow}

\newcommand{\remove}[1]{}
\newcommand{\oLog}{{\mathcal O}^{Log} } 
\newcommand{\cc}[1]{\mathcal{#1}}

\makeatletter
\newcommand{\mathleft}{\@fleqntrue\@mathmargin0pt}
\makeatother

\begin{document}
%

%

%

\title{Secure Logging with Security against Adaptive Crash Attack \thanks{A short version of this paper has been accepted to the 12th international symposium on Foundations and Practice of Security (FPS 2019).}}

\author{Sepideh Avizheh, Reihaneh Safavi-Naini, Shuai Li}
\institute{University of Calgary, Alberta, Canada}

\maketitle
%

	\begin{abstract}

	Logging systems are an essential component of security systems and their security has been widely studied. Recently (2017) it was shown that existing secure logging protocols are vulnerable to {\em crash attack } in which the adversary modifies the log file and then crashes the system to make it indistinguishable from a normal system crash. The attacker was assumed to be non-adaptive and not be able to see the file content before modifying and   crashing it (which will be immediately after modifying the file). The authors also proposed a system called SLiC that protects against this attacker. In this paper, we consider an (insider) adaptive adversary who can see  the file content as new log operations are  performed. This is a powerful  adversary who can attempt to rewind the system to a past state. We formalize security against this adversary and introduce a  scheme with provable security. We show that security against this attacker requires some (small) protected memory that can become accessible to the attacker after the system compromise. We show that existing secure logging schemes are insecure  in this setting, even if the system provides some protected  memory as above. We propose a novel  mechanism that, in its basic form,  uses  a pair of keys that evolve at different rates, and employ this mechanism  in an existing logging scheme that has forward integrity to obtain	a system with provable security against adaptive (and hence non-adaptive) crash attack. We implemented our scheme on a desktop computer and a Raspberry Pi, and showed in addition to higher security, a significant efficiency gain over SLiC.
		
\end{abstract}

%
%


\section{Introduction}
Computer systems use logging function to store and keep track of important events  in the system.
Log files are used for a variety of purposes including trouble shooting, intrusion detection and forensics \cite{logTroubleShooting,logAccountability,logForensic}.
In many cases, adversaries want to stay covert and be able to modify the log files without being detected. Thus, integrity of log data is essential, and
protecting the log files against tampering and modification has been an active area of research. The simplest  form of protection is to store each log entry  with the corresponding message authentication code (MAC), and with a key that is unique to the entry to ensure the entries cannot be permuted \cite{BY}.
One  cannot expect any protection after the time of system compromise: the attacker is assumed to have access to the system, algorithms, and the keys at the compromise time (full state of the system) and can add any log that they desire afterwards. 
Thus the goal of protection is maintaining integrity  of the past logs. This is called {\em forward security} or {\em forward integrity} \cite{BY} and is achieved by evolving forward (using a one-way function) the key that is used for generating integrity information of log entries. In
\cite{Schneier-Kelsey}, authors used forward integrity based one MAC and hash chains and proposed a secure audit log for a local untrusted logging device that has infrequent communication with the verifier. LogCrypt \cite{Logcrypt} made some improvement to \cite{Schneier-Kelsey} such as the ability to use public key cryptography as well as aggregating multiple log entries to reduce latency and computational load. Forward integrity, however, does not protect against truncating of the log file: the adversary can remove entries at the end without being detected.  This attack can be protected against by including an aggregate MAC (signature) to the file, which is proposed in \cite{FSSAA,MaTsudik} through the notion of forward-secure sequential aggregate authentication.  The authentication data (tag) per log entry is sequentially aggregated and the individual tag is removed.
 
When a system crash happens, the data that are stored in caches (temporary memories) will be erased or become unreliable.
Caches may include new log entries and updates to the stored log entries, so a crash would result in the loss of new entries, that have not been stored yet,  and inconsistency of existing ones.  This provides a window of opportunity  for attackers to modify the log file and remain undetected by crashing the system.
Crash attack was introduced and formalized  by Blass and Noubir \cite{CrashAttack}. They showed that all existing secure log  systems
were  vulnerable to this attack.
 Blass and Noubir  formalized the security notion of crash integrity using a game between  the adversary and a challenger, and  proposed a system, SLiC, that provides protection in this model.
 SLiC encrypts and permutes the log entries so that they cannot be known to a non-adaptive adversary who gets only one time read and tampering access to the system (please see Appendix \ref{SLiCDiscription} for details).

\vspace{1mm}
\noindent
{\bf Our work:}  We consider a  secure logging  system that  uses an initial key (that is shared with the verifier) to generate authenticated log entries  that are stored in the log file.
We assume an (insider) {\em adaptive  crash adversary }who can   adaptively choose the messages that will be logged and can 
see the log file after each logging operation. The goal of the adversary is to remove and/or  tamper with the logged elements.
We show that  without  other assumptions and by the verifier only using their  secret key, it is impossible to provide security against adaptive crash attack.  We thus assume the system stores (and evolves)  its keys in a small protected memory, that will become 
accessible to the adversary after the system is compromised.
 Such a memory can be implemented using trusted hardware modules whose content will not be observable during the normal operation of the system, but can become accessible if the system crashes. We formalize security and show that  SLiC is  insecure in this model and an adversary who can see the intermediate states of the log file can  successfully {\em rewind} the system to a previous state. 

{\em Adaptive crash resistance:}  We introduce a {\em double evolving key mechanism} which, in the nutshell,
uses two keys, one evolving with each log event and one evolving at random intervals, that 
reduces the success chance of crash attack even if the adversary is adaptive.
 The keys become available after the system compromise but the random interval evolution limits the success probability of the adversary to successfully  rewind the  system to a previous state.
We analyze this system in our proposed model and prove its security against an adaptive attacker. This mechanism can be extended to multiple independent keys evolving at different rates to enhance the security guarantee of the system.
We implemented double evolving key mechanism on a windows PC and Raspberry PI and compared the results with those reported for SLiC  \cite{CrashAttack}, showing significantly improved time-efficiency.

{\em Discussion:}  The double evolving key mechanism keeps the logged events in plaintext and  provides an elegant and very efficient solution against non-adaptive crash attack. SLiC, the only secure logging system with security against (non-adaptive) crash attack, provides security by encrypting and permuting elements of the log file.  This makes access to logged data extremely inefficient: one needs to reverse the encryption and permutation to access the required element.
For functionalities such as searching for a pattern or keyword, this means recovering the whole log file which is impractical.
The comparison of the two systems  is further  discussed in Section \ref{SecurityAnalysis}.

{\em  Organization:} Section \ref{Background}  gives the  background; Section \ref{SystemModel}, describes adaptive crash model  and its relation to non-adaptive case. Section \ref{RewindSecure} proposes the double evolving key mechanism and Section \ref{SecurityAnalysis} is on the security and complexity analysis of our scheme.
 Section \ref{Implementation} explains the implementation, and Section \ref{Conclusion} concludes the paper.

\section{Preliminaries}\label{Background}

We use the system model of Blass et al. \cite{CrashAttack} which models 
many systems that are used in practice, 
and focus on the settings where the verifier is mostly offline and checks 
 the log file once in a while 
 (infrequently).  

An {\em event }$m_i$ is a bit string that is stored in the log file together with  an authentication tag $h_i$, such as $h_i= HMAC_{k_i}(m_i)$.
The key $k_i$ is for authentication of the $i^{th}$ log entry and is generated from an initial seed. The key $k_i$ is evolved to $k_{i+1}$ for $(i+1)^{th}$ entry and $k_i$ is removed.
Using a different key for each element protects not only against reordering, but also ensures that if the key is leaked, past keys cannot be obtained and past entries cannot be changed. A common way of evolving a key is by using a  {\em pseudorandom function family $PRF_k(.)$} indexed by a set of keys \cite{BY}, that is, $k_{i+1} = PRF_{k_i} (\chi) $, where  $\chi$ is a constant. The security guarantee of a PRF family, informally, stated as follows:
 a function that is chosen randomly from the PRF family  cannot be distinguished from a random oracle (a function whose outputs are chosen at random), using an efficient algorithm, with significant advantage. 
To   protect against {\em truncation attack} where the adversary removes the last $t$ elements of the log file, one can add an aggregate hash  $h_{i+1}= HMAC_{k_{i+1}}(m_{i+1},  h_i )$ and delete $h_i$, or use an {\em aggregate signature} where signatures generated by a single signer are sequentially combined. If verification of aggregate signature is successful, all the single signatures are valid; otherwise, at least one single signature is invalid.

In all these schemes, the event sequence order in the log file remains the same as the original event sequence, and the verification requires only an original seed from which the key for the rest of the system can be reconstructed.
 Crash attack uses this property and the fact that a crash will remove all the new events and a number of the stored events that must be updated, so it makes parts of the log file, including the stored keys, inconsistent.
 This possibility in a crash can be exploited by the adversary to launch a successful truncation attack. Blass et al. system, called SLiC \cite{CrashAttack}, protects against crash attack by encrypting  each stored  log entry (so makes them indistinguishable from random), and  uses a randomized mapping that 
 permutes the order of the log entries  on the log file using a pseudorandom number generator (PRG). 
 Informally, a PRG uses a seed to  generate  a sequence of numbers that is
   indistinguishable from a random sequence.
Using the PRG, the order of storing events in the log file will appear  ``random" to the adversary who does not know the PRG seed and so truncation attack is prevented.
 This protection however will not work against an adaptive attacker who will be able to see the result of storing a new event, and  by comparing the new re-ordered  log file with the previous one learn the places that can be tampered with (See section \ref{AdaptiveGame} for details of the attack).

As outlined above, storing a new event and its authentication  data will result
in the  update of some existing entries in the log file. In particular, to update  a stored value $x$ to $x'$, the following steps will happen:
(i) read $x$ and compute $x'$, (ii) store $x'$, and  (iii)  delete $x$.  However the last two steps may be re-ordered by the operating system, so when a crash happens,  the state of the update will become unknown: that is $x$ has been deleted and $x'$ has not been written yet.
This reordering would result in inconsistency during the log verification.
When a crash happens, the data in the cache becomes unreliable and the verification of the log file requires not only the initial seed, but also an estimate of the part of the log file that is verifiable.
 \remove{This is determined using the  parameter {\em cache size} \cite{CrashAttack}, and estimating} Similar situation can happen in the update of keys, resulting in both $k_{i-1}$ and  $k_{i}$  to become unavailable for  the system recovery. The goal of the verifier is to recover the largest  verifiable log sequence from the crashed system.

	\vspace{-0.4cm}
	\section{System and adversary model}\label{SystemModel}
	We first give an overview of our system and the adversary model. 
	There are three entities: 1) a logging device $\cc{L}$, 2) a verifier $\cc{V}$, and 3) an adversary $\cc{A}$. 
	\remove{
	The logging device  receives 
	{\em  log events $m_i\in \{0,1\}^*$} that are stored in a permanent  storage medium (log file), that we refer to as  {\em disk}. Storing a  log event requires	first the update of the systems'  keys,  followed by computing the required authentication data, and finally generation of
	 a set  of write operations	$\{o(m_u)  \cdots o(m_v)\}$  on the disk.
	We assume that there is a small protected non-volatile  memory in the system that will be used for system's current keys. We refer to this non-volatile memory as {\em key disk}.
	$\cc{A}$  is an insider who can see the log file updates. They can also compromise the system and after that  obtain full  access to all the system log file, key disk, and the cache. 
	The adversary uses purposeful crash to make the compromise and modifications indistinguishable from a normal crash.
	That is we assume  a crash will make the protection of  the key disk ineffective and the adversary will gain full access to the whole system state.
	The verifier $\cc{V}$ has the system's initialization key, and their goal is to recover the largest verifiable log sequence from the crashed state.
The security goal  is to ensure that  $\cc{A}$ cannot modify or delete  events from the log file, after the crash.}

	\textbf{Logging device} $\cc{L}$, stores the event  sequence and  current keys  using the following types of storages:
		(i) LStore is a disk (long term storage)  that stores	 log events.  This disk can be read by the attacker  when the system is compromised.
		(ii) Log cache  is a temporary memory that is used for  
		the update of the LStore.
		(iii)  KStore is the key disk  that is used to store current keys of the system.  This is a non-volatile memory that will become available to the adversary when the system crashes. 
		KStore   uses a protected cache for its update. \\	
	The logging device receives a sequence of events $m_1,m_2 \dots  m_i, i\in N^+$, where   $i$ is the order of appearance of the event $m_i$ in the sequence, and $N^+$ is the set of positive integers. 
	The state of the logging device  after $m_i$ is logged,  is specified by $\Sigma_i=[\Sigma^{K}_{i}, \Sigma^{L}_{i}, Cache_i]$, where 
	$\Sigma^{L}_{i}$,  $\Sigma^{K}_{i}$ and  $Cache_i=\{cache^{L}_{i},cache^{K}_{i}\}$ are the states of  the LStore,   the KStore  and their caches, respectively, after  $m_i$ is logged. 
	
	The {\em log operation} $Log(\Sigma_{i-1}, m_i)$ takes the   state 
	$\Sigma_{i-1}$, and the  log event  $m_i$, 
	uses the cache as a temporary storage, and updates  LStore  for  the storage of  the  (processed) log event. 
	 This operation uses KStore cache to update the keys in the KStore. We assume this cache only holds 
	the required data for updating  $k_{i-1}$ to $k_i$ that is used in $Log(\Sigma_{i-1}, m_i)$. This assumption is used to estimate the amount of key information that will be unreliable after a crash.  We also assume that KStore has enough size to hold the current key $k_{i}$.
 	
	The $Log(\Sigma_{i-1}, m_i)$ operation,  (i) generates  a set  of write operations\\ $\{o(m_u)  \cdots o(m_v)\}$, which we denote with $\oLog(\Sigma_{i-1}, m_i)$,  on the LStore (i.e. $\Sigma^{L}_{i}$ and its associated cache are updated), and (ii) updates KStore (i.e. $\Sigma^{K}_{i}$ and its cache are updated).
A disk write operation $o(m_i)$ (we call it a log file entry)  writes to the disk $m_i$ together with its authentication data.
The initial states of LStore and KStore  are  denoted by $\Sigma^{K}_{0}$ and $\Sigma^{L}_{0}$, respectively. $\Sigma^{L}_{0}$ contains an initial event that is used to detect complete deletion of the disk. $\Sigma^{K}_{0}$  contains the initial keys of the system.
	 
	As log events are processed, the states of the two storage systems will be updated in concert:
	after $n$ log operations, the length of  $\Sigma^{L}_{0}$ is  $n$, and the length of  $\Sigma^{K}_{0}$ is unchanged, but the content has been updated to the new values.
	The initial state of the system $\Sigma_0$ will be securely stored  and later used for verification. 
\remove{
\noindent	
	{\em Normal crash:} In a normal crash the log events which have not yet been written to the LStore
	 will be removed, and the log entries that are to be updated may be removed or become inconsistent with their authentication data.
	 This is because  of possible re-ordering of the update  sequence, (i) move $m_{i-1}$ to cache and update to $m_{i}$,  (ii)  write $m_i$
	 to LStore, and iii) delete $m_{i-1}$ LStore.
A similar situation can happen in the update of   KStore  resulting in both $k_{i-1}$ and  $k_{i}$  to become unavailable for  the system recovery.}


	\textbf{Adversary},
	\remove{During a system crash however the LStore and KStore caches  become unreliable, and there is a chance that  the key values that are written to the disk become inconsistent because crash happens at the time of key update. 
	The goal of the crash adversary is to modify the LStore and KStore  such that a verifier who uses the initial state of the system, and the crashed state,	cannot detect the attack from the events that are recovered  from  the LStore.}
	$\cc{A}$, (i) adaptively generates events  that will be processed by the $Log(\cdot,\cdot)$ operation of $\cc{L}$;
	 $\cc{A}$ can see  LStore and its cache after each $Log(\cdot,\cdot)$ operation; (ii) compromises $\cc{L}$  and accesses  KStore and its cache,  and
	modifies the state of $\cc{L}$, and finally crashes  the system.
	The goal of the crash adversary is to modify the LStore and KStore  such that a verifier who uses the initial state of the system, and the crashed state cannot detect the attack. In Section \ref{AdaptiveGame}, we define our security game using this model. 
\vspace{-0.7cm}
\begin{figure*}
	\centering
	\setlength{\abovecaptionskip}{0ex}
	\setlength{\belowcaptionskip}{-4ex}
	\includegraphics[width=0.9\textwidth]{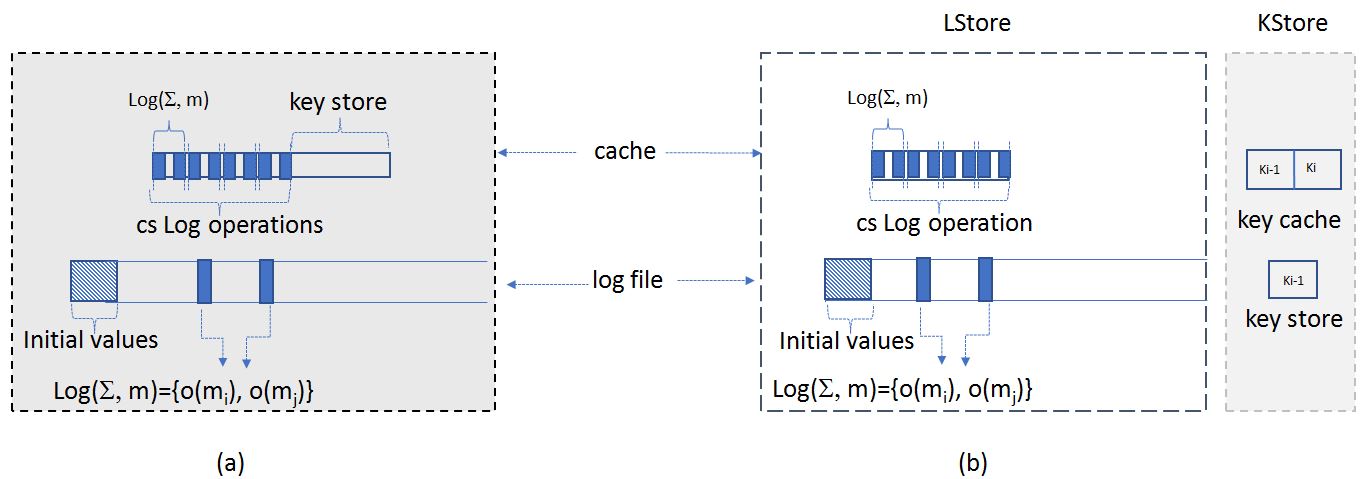}
	\caption{\small (a) Non-adaptive adversary, \remove{:  Log operation generates write events $o(m_i),\cdots, o(m_j)$; cache stores write operations of $cs$  logging operations, the current keys and other system parameters. \remove{In the security game the adversary does not see the intermediate results of logging operation.} After the system compromise, the log file and cache will become available to the adversary.}
	(b) Adaptive adversary. \remove{: Log operation is the same as (a);  cache stores write operations of $cs$ logging operations.  A key-store, KStore, stores the key that is required for one disk write $o(m_i)$, and a key cache is used for update of this key.
	In the security game the adversary sees intermediate results  of  the logging operation; after the compromise, the content of the key store and key cache will become available to the adversary.
	\remove{In both systems the initial state of the system, consisting of the initial key, and the initial value of the log file will be used in Recover function.}} The shadowed parts are invisible to adversary before compromising the system.} \label{Models}
\end{figure*}

Figure \ref{Models} shows the differences between our adversarial model and that of Blass et al. \cite{CrashAttack}. In Figure \ref{Models}.(a) $Log(\cdot,\cdot)$ operation generates disk writes in the cache first, which are then written to the log file. 
The system current key resides in the system cache also. The adversary can use the $Log(\cdot,\cdot)$ operation on the message sequence of their choice but cannot see the intermediate results of logging until the system is compromised. It is easy to see that in this model it is impossible to provide security if the adversary is given access to the system after each $Log(\cdot,\cdot)$ operation:  the adversary observes the current key and can simply use it to generate any arbitrary log event and later write it in the log file without being detected. Figure \ref{Models}.(b) shows our model. 

	We will not consider the case that the adversary adds new entries to LStore: this can always be  done  undetectably because the adversary knows the content of the KStore after the compromise.  We however require that the log events that have been stored ``before the time  of the compromise (crash)'', remain untouched.

	\textbf{Verifier}, $\cc{V}$, uses $Recover(\cdot,\cdot)$  algorithm that takes  the current state of the logging system, and the initial state $\Sigma_0$,	and outputs either the list of consistently stored events, or $\perp$ which indicates untrusted log.
	
	\remove{We use a parameter $cs$, first introduced in \cite{CrashAttack}, that  is the maximum number of log events that will be lost  during a  normal crash.
	This number can be estimated for a particular implementation (e.g., taking into account the caches of  OS, file system, hard disk, ...),
	and allows us to estimate the maximum length of the LStore.	More specifically, a $Log(\cdot,\cdot)$ operation  will include  at least one $o(m_i)$, and  so will affect at least one element of LStore.
	  Thus if  the highest recovered index after a crash is $u_1$, one may assume that  at most $m_{u_1+1}\dots m_{u_1+cs}$ have not had the chance of being written to the LStore.  
	That is, the disk writes of these events,  $\oLog(\Sigma_{u_1}, m_{u_1+1})$ $\cup \dots \cup$ $\oLog(\Sigma_{u_1+cs-1}, m_{u_1+cs})$,   
	 have not been completed, and the corresponding cells in the LStore have not been correctly updated.}
	 \vspace{-0.2cm}
	\subsection{Logging Protocols}
	A logging  protocol $\Pi$ consists of three algorithms: 
	
	\noindent	1.  $Gen(1^{\lambda})$: 
		$Gen$ takes a security parameter $\lambda$ and outputs $\Sigma_0$, which is $\cc{L}$'s initial state, $\Sigma_0=[\Sigma^{L}_{0},\Sigma^{K}_{0},Cache_0=\emptyset]$, and will be  stored  securely for future use by the verifier $\cc{V}$.
	The initial state includes: (i) $\Sigma^{L}_{0}$, that is, the initial state of the log file and it is initialized securely to protect against complete removal of the log file, (ii) $\Sigma^{K}_{0}$, that  stores the initial seed keys,		and (iii) $ Cache_0=\{ cache^{L}_0, cache^{K}_{0}\}$, which are assumed  to be initially empty.
	
	\noindent	2. $Log(\Sigma_{i-1},m_i)$:  Let $\Sigma_{i-1}=[\Sigma^{K}_{i-1}, \Sigma^{L}_{i-1},  Cache]$ be the current state after $i-1$ sequence of events are logged. For an  event $m_i\in\{0,1\}^{*}$,  and the current state $\Sigma_{i-1}$, the operation  $Log(\cdot,\cdot)$ outputs, either a new state $\Sigma_i$, or a special state $\Sigma^{cr}_{i}$, called a {\em crashed state}.
		A non-crashed state is a {\em valid state} that is the result of using $Log(\cdot,\cdot)$ consecutively on a sequence of log events.
		If  $Log(\cdot,\cdot) $  outputs a crashed  state,  the device $\cc{L}$ has been crashed and needs to be initialized.
		
\noindent	3.  $Recover(\Sigma, \Sigma_{0})$: Receives an initial state $\Sigma_0$ and   a 	(possibly crashed) state $\Sigma$, and verifies if it is an untampered state that has resulted from $\Sigma_0$ through consecutive invocation of $Log(\cdot,\cdot)$. 
		$Recover(\Sigma, \Sigma_{0})$ reconstructs the longest sequence of
		 events  in	the LStore that pass the system integrity checks, or outputs $\perp$ which indicates an untrusted log.
		 If $\Sigma$ had been obtained from $\Sigma_0$ by consecutive applications of $n$  $Log(\cdot,\cdot)$,  then
		$Recover(\Sigma, \Sigma_{0})$ will output the $n$ logged events.
		Otherwise the set, $\cc{R}$,  of  recovered events consists of	$n' < n$  pairs	$\cc{R}=\{(\rho_{1},m^{\prime}_1),...,$ $(\rho_{n'},m^{\prime}_{n'})\}$.
		If one of  $m^{\prime}_j \neq m_{\rho_j}$, the adversary has been able to successfully modify a log entry.
		 For example the correct pair with $\rho_i=4$ will have $m^{\prime}_i = m_4$.  We use  $n$ and $n^{\prime}$ to denote the length of the logged sequence  before crash, and  the highest index  of the log file  seen by $Recover(\cdot,\cdot)$.		
		The input state $\Sigma$ to   $Recover(\cdot,\cdot)$  can be:
		(i) a valid state of the form $(Log(Log(...Log(\Sigma_{0},m_{1})...),m_{n})$, so  $Recover(\cdot,\cdot)$  outputs, $\{(1,m_1),...,(n,m_n)\}$;
			(ii) a state which is the result of a normal crash, so $Recover(\cdot,\cdot)$ outputs,$(\rho_1, m'_1), \cdots (\rho_{n'}, m'_{n'})$ where $n' < n$;
			(iii) a state which is neither of the above, so $Recover(\cdot,\cdot)$  outputs $\perp$, and a modified (forged) or missing  event is detected. 
	

\textbf{Efficiency}: To support high frequency logging and resource constrained hardware, $Log(.,.)$ is required  to be an efficient algorithm. 
	
	\vspace{-0.3cm}
	\subsection{Cache}
	 We use (a parameter) cache size $cs$, first introduced in \cite{CrashAttack}, to estimate the effect of crash when recovering the log file. $cs$ is the maximum number of log events that will be lost  during a  normal crash.
	This number can be estimated for a particular implementation (e.g., taking into account the caches of  operating system, file system, hard disk, ...),
	and allows us to estimate the maximum length of unreliable log events.
	\remove{
	More specifically, a $Log(\cdot,\cdot)$ operation  will include  at least one $o(m_i)$, and  so will affect at least one element of LStore.
	Thus if  the highest recovered index after a crash is $n'$, one may assume that  at most $m_{n'+1}\dots m_{n'+cs}$ have not had the chance of being written to the LStore and $m_{n'-cs+1}\dots m_{n'}$ have not been completely written.}

	Logging the event $m_i$  will generate a set of disk write operations, $\oLog(\Sigma_{i-1}, m_i)=\{o(m_u),...,o(m_v)\}$, that will add a new entry to the LStore and may update a number of other entries.  
	The log  operation $ Log(\Sigma_{i-1}, m_i)$ will also update     
	$k_{j-1}$ to $k_j$, for all  $o(m_j) \in \oLog(\Sigma_{i-1}, m_i)$.
	If $\cc{L}$ crashes  before $Log(\Sigma_{i-1},m_i)$  completes,	all $o(m_j)\in \oLog($ $\Sigma_{i-1},m_i)$ will be lost. This is because all these operations are 
	 in cache.
	For simplicity, we assume the KStore stores the key $k_j$ which is used in constructing $o(m_j)$ only.
	To perform $Log(\Sigma_{i-1},m_i)$, each  $o(m_j)\in \oLog($ $\Sigma_{i-1},m_i)$ will be 
	processed once at a time (The argument can be extended to the case that KStore  is larger).
	If crash happens, the  $k_j$ that is being updated will also become unreliable.
	The notion of expendable set, first introduced in \cite{CrashAttack}, captures the  $LStore$  entries that are considered unreliable when a crash happens.
	
	\textbf{Definition 1} (Expendable set ($ExpSet$)).  Let $\Sigma_n$ be a valid state comprising events $\{m_1,...,m_n\}$, and $Cache_n = \emptyset$.  Let $Cache_{n'}$ be the content of  cache after $\cc{L}$ adds events $(m_{n+1}, ... , m_{n'})$ using the $Log(\cdot,\cdot)$ operation. An event $m_i$ is expendable in state $\Sigma_{n'}$, iff
	$(o(m_i)\in \{O^{Log}(\Sigma_n, m_{n+1})\cup \dots \cup O^{Log}(\Sigma_{n'-1}, m_{n'})\}) \wedge (o(m_i)\in Cache_{n'}$). 
	The set of all expendable log entries in  $\Sigma_{n'}$ is denoted by
	$ExpSet $.  
	
	 The  definition identifies 
	  $o(m_i)$s that are in   
	   the expendable set assuming 
	 the 
	 first and the last state of the cache are known. 
	In practice  however, the verifier  
	 receives a log file of size $n'$ (events)  and without  knowing the final state of the system must decide  
	 on the length of the file that has reliable data.
	 If the cache  can hold $cs$ events, then  we consider $2cs$  events (the interval $[n- cs+1, n+cs] $) as expendable set. 
	 This is the set of events who could have resided in the cache when the crash occured.
	 Note that logging an event may generate  more that one disk write operation  that could be update of the earlier entries in the log file.
	%
	The following proposition  summarizes the discussion above. 
	\begin{proposition}(Determining expendable set). Let $\Sigma_{n'}$ be the state of the system after logging  $m_{1}, \cdots, m_{n'}$. 
		An event $m_i$ is expendable in a  state $\Sigma_{n'}$, where $n'$ is the highest index of a log entry in the LStore\footnote{Note that this LStore may be the result of normal logging operation, or after a crash.},
		 if $o(m_i)\in \{O^{Log}(\Sigma_{n'-cs},m_{n'-cs+1})$ $\cup \dots \cup O^{Log}(\Sigma_{n'+cs-1}, m_{n'+cs})\} \mbox{ and }$ possibly $
		o(m_i)\in
		Cache_{n'}. $ The set of all expendable log entries in  the recovered state $\Sigma_{n'}$ is
		$ExpSet 
		=\{m_i: m_i \mbox{  is expendable   in }  \Sigma_{n'}\} $.
	\end{proposition}
\vspace{-0.55cm}
	\begin{proof}
	We assume  cache will include up to $cs$ log events.  These events,  
	(i)  may all  be  events after $n'$; that is, from $Log(\Sigma_{n'}, m_{n'+1})$ $\cup \dots \cup$ $Log(\Sigma_{n'+cs-1}, m_{n'+cs})$, events $[(n'+1,o(m_{n'+1})),\dots,(n'+cs,o(m_{n'+cs}))]$ may have been lost,  and other disk write  events may not have been completed, or (ii) the writing is incomplete, so the logging of up to $cs$ events before $n'$ will have incomplete disk write and $Log(\Sigma_{n'-cs}, m_{n'-cs+1})$ $\cup \dots \cup$ $Log(\Sigma_{n'-1}, m_{n'})$, have been damaged, or (iii) a  random set of  $cs$ events in $Log(\Sigma_{n'-cs}, m_{n'-cs+1})$ $\cup \dots \cup Log(\Sigma_{n'+cs-1},$ $ m_{n'+cs})$, have been lost. Therefore, all the log events in the  range $[n'-cs+1,n'+cs]$ is considered to be  expendable set. 
	Thus the size of expendable set on the LStore is  $2cs$ events.

	\remove{When $n'$ is the  highest seen index of a crashed state in LStore,
	the crash may have happened when logging any  event  with index $n'-cs+1\cdots n'+cs$ in the original  log sequence and so   keys $\{K_{n'-cs+1},\cdots , K_{n'},  K_{n'+1}, \cdots  K_{n'+cs} \}$ become unreliable, where $K_i$ is all the key information required for a single log operation $Log(\Sigma_{i-1},m_i)$ and includes all the keys $k_j$ corresponding to one writing operation $o(m_j) \in \oLog(\Sigma_{i-1},m_i)$.
	The remaining key entries  in the KStore  are expected to be reliable and  must be consistent with LStore content. } 
	\end{proof}
	\vspace{-4mm}
	\subsection{Security Definition} \label{AdaptiveGame}
	The effect of crash on the system in general depends on the hardware, and is abstracted by the cache size parameter $cs$.
Our new security definition for adaptive crash attack is given in Algorithm  \ref{Game2}.  
We define a security game between the challenger and  an adversary   $\cc{A}$ that has access to the following oracles.
	
	\textbf{Gen oracle:} $GEN_Q()$ allows the adversary $\cc{A}$ to initialize a log on $\cc{L}$.  $\cc{C}$ runs $Gen(1^{\lambda})$ and returns the initial state of LStore $\Sigma'^{L}_{0}$ and its associated cache $cache'^{L}_{0}$. The state $\Sigma'_0$ is stored in the set $Q$ that  records the log queries made by the adversary.
	
	\textbf{Log oracle:} $LOG_{\Sigma,Q}()$, is a stateful function,   which allows the adversary	$\cc{A}$  to  adaptively  log events on  $\cc{L}$: the adversary adaptively chooses a message $m$ to be logged, $\cc{C}$  runs  $Log(\cdot,\cdot)$ using the current state $\Sigma$ of the system, and returns,  $\Sigma'^{L}$ (state of the LStore) and the $cache'^L$  (state of the cache)  to $\cc{A}$. 
	The state $\Sigma'$ is stored in the set $Q$ that  records the log queries made by the adversary (This is later used to detect rewind attack).
	
	\textbf{Recover oracle:} $REC_{\Sigma}()$, is a stateful function, that can be called in any state by $\cc{A}$. To respond, $\cc{C}$ runs $Recover(\Sigma,\Sigma_0)$ and returns the recovered set $\cc{R}$ which can be either $\perp$ or $\{(\rho_1,m'_1),\cdots,$ $(\rho_{n'},m'_{n'})\}$.
	
	\textbf{Crash oracle:} $CRASH_{\Sigma}()$, is  a stateful function, that can be called by $\cc{A}$ on any state $\Sigma$ and  allows  $\cc{A}$  to learn the effect of crash on the system by accessing  the complete state $\Sigma$ of the system including the KStore. $CRASH_{\Sigma}()$ returns $\Sigma^{cr}$ as the state of the logging device $\cc{L}$.
	
	In Algorithm \ref{Game2}, the first stage is for adversary to learn. $\cc{A}$ gets oracle access to all the functions mentioned above and chooses $n$ messages to log.  Challenger $\cc{C}$,  generates the initial keys and initializes the $KStore$, $LStore$, and the $Cache$. 
	At this stage, adversary has oracle access to $GEN_Q()$, $LOG_{\Sigma,Q}()$ and $CRASH_{\Sigma}()$.	$\cc{A}$ adaptively issues $n$ log queries, $m_1 \dots m_n$, to $LOG_{\Sigma,Q}()$ oracle.
	The oracle executes  $Log(\Sigma, m)$ for each message and returns  the  LStore and $ cache'^L$  of the resulting state $\Sigma'$ to adversary. 
	$\Sigma'$ is   stored in the queried set $Q$. After $n$ calls to $Log(\cdot,\cdot)$, $\cc{A}$  calls $CRASH_{\Sigma}()$, gets full access to the LStore, KStore and Cache, which  all will  be tampered as desired, and then crashes the system.
	Adversary outputs a sequence of $\ell$ positions $\alpha_i$, where $\alpha_i \in [1,n]$, none of which correspond  to the   index of an element in the  expendable set, assuming $n$ is the highest index in LStore seen by the verifier. The algorithm $Recover(\cdot,\cdot)$ outputs a sequence of $n'< n$ index-event pairs $\{(\rho_i, m'_i)\}$.
	Intuitively,  the adversary wins if, (i) one of their outputted indexes appear in $ \cc{R}$ with a value different from the original logged sequence  (i.e. successfully  changed by the adversary), (ii) one of the outputted indexes does not appear in $ \cc{R}$ (that is successfully deleted by the adversary),  or (ii)  the  recovered list $ \cc{R}$ matches the LStore of one of the queried states.
 
 	\textbf{Definition 2} (Crash Integrity). $\cc{A}$ logging protocol $\Pi=(Gen,Log,Recover)$ provides $f(\lambda)$-crash integrity against adaptive adversary $A$, iff for all PPT \footnote{Probabilistic Polynomial Time} adversaries there exist a negligible function $f(\cdot)$ such that $Pr[Exp^{AdapCr}_{A,\Pi}(1^\lambda,cs)] \leq f(\lambda)$.

	\subsection{Impossibility result}
	Existing secure log schemes, i.e. \cite{BY,Schneier-Kelsey,Logcrypt,MaTsudik}, consider an ordered log where a new log entry is appended to the end of LStore. 
	These schemes use key evolution but do not use secure hardware or platforms to store  the latest secret key that captures the state of the log file.  Nor do they rely on a trusted third party to safeguard this information.
These protocols are vulnerable to non-adaptive crash attack \cite{CrashAttack}  because adversary knows  the order of log entries,
 can truncate the log file  
and delete the keys, leaving the system in a   stateless situation, which 
 makes 
it  impossible to distinguish a crash attack from a normal crash.
 SLiC, is the only known crash tolerant scheme \cite{CrashAttack} which masks the order of elements in the log file by
 encrypting them and applying 
 a random permutation on the location of log entries in the LStore. However, it cannot protect against rewinding in an adaptive adversarial model. 
All existing schemes,  
 including SLiC, are vulnerable to adaptive crash attack even considering a protected KStore according to our model. This is because
  the KStore 
  can be undetectably removed or modified when the system is  compromised and this will again put the logging system in a state that is indistinguishable from a normal crash.
  In another words,  a logging system that cannot \textit{reliably} protect its state information during logging operation, 
  and assuming an adaptive adversary who can see the LStore, is  subjective to rewinding.
	
	We note  that \textit{$Exp^{AdapCr}_{\cc{A},\Pi}()$} is stronger than  \textit{$Exp^{CrInt}_{\cc{A},\Pi,Crash}()$} game \cite{CrashAttack}. This can be proved by showing two claims.
	Claim 1: if a non-adaptive adversary $\cc{A}_{na}$ is successful in breaking a scheme, 
	an adaptive adversary $\cc{A}_{a}$ will also succeed with at least the same probability.
	This is true because 
	 if there exists $\cc{A}_{na}$ that  can rewind the system 
	  to a previous state, 
	  and claim a normal crash, $\cc{A}_{a}$ can use it as a  subroutine to break the scheme.
	   This implies that all existing schemes \cite{BY,Schneier-Kelsey,Logcrypt,MaTsudik} are vulnerable to adaptive crash.
	Claim 2:  SLiC that is secure against a non-adaptive adversary cannot protect against rewinding. 
	SLiC encrypts and permutes the log events randomly. However because the adversary can 
	see the content of 
	LStore after each operation,  this permutation will be reversible. 
	$\cc{A}_{a}$ can rewind the log file to any past state that it has already seen (which has a different but valid permutation), 
	remove the keys from KStore (causing the system to be stateless) and claim a normal crash.
	These two arguments are formalized 
	in   Appendix \ref{Impossibility}.
	
\begin{scriptsize}
	\begin{algorithm}[!t]
		\scriptsize
		\caption{ $Exp^{AdapCr}_{\cc{A},\Pi}(1^{\lambda},cs):$}\label{Game2}
		\begin{algorithmic}[1]
			\State $(m_1,m_2,\dots,m_{n}) \gets \cc{A}(1^{\lambda},GEN_Q(),LOG_{\Sigma,Q}(),REC_{\Sigma}(),CRASH_{\Sigma}())$		
			\State $\Sigma_{0} \gets Gen(1^{\lambda})$ \newline \indent //$\Sigma_0=[\Sigma^{K}_{0},\Sigma^{L}_{0},Cache_0]$
			\State $(\Sigma^{cr},\alpha_1,\dots, \alpha_\ell) \gets \cc{A}^{GEN_Q(),LOG_{\Sigma,Q}(),CRASH_{\Sigma}()}({\Sigma^{L}_0,m_1,\dots,m_{n}})$ \newline \indent //$\Sigma^{cr}=[\Sigma^{cr,L}, \Sigma^{cr,K},Cache^{cr}]$, $\alpha_i \in [1,n]$
			\State $\cc{R} \gets Recover(\Sigma^{cr},\Sigma_0)$
			\newline \indent $//\cc{R}=\perp \quad or \quad \cc{R}={(\rho_1,m'_1), \dots, (\rho_{n'},m'_{n'})}$
			\If{$\cc{R}=\perp$}
			\State Output $\perp$
			\ElsIf{ \newline \indent $[ \exists (\alpha_i,\rho_j): (\alpha_i=\rho_j) \wedge (m_{\alpha_i} \neq m'_j) ]  \vee \quad //Modify \newline \indent
				[\exists \alpha_i \notin  ExpSet: 
				\rho_j \neq \alpha_i, \forall j=1, \dots, n' ] \vee \quad //Delete \newline \indent
				[\cc{R}=   \Sigma'^L, \Sigma'\in Q] \quad //Rewind
				$
				\State Output \textit{Success}}
			\EndIf
			\State \textbf{end} \newline \indent

			\begin{tabular}[h]{l l}
				
				\begin{tabular}{@{}l@{}}
					//\mbox{$GEN$ runs $Gen(\cdot)$, returns $\Sigma_0$} \\
					\indent
					\textbf{$GEN_Q():$} \\ \indent
					$\Sigma'_0 \gets Gen(1^\lambda)$ \\ \indent
					$Q \gets Q \cup  \Sigma'_0$\\ \indent
					Return $(\Sigma'^{L}_{0},cache'^{L}_{0})$ \\ \indent 
				\end{tabular}
				
				&
				
				\begin{tabular}{@{}l@{}}
					//\mbox{$LOG$ runs $Log(\cdot,\cdot)$ on $m$,} \\ \mbox{ returns the state  of LStore and Cache} \\ \indent
					\textbf{$LOG_{\Sigma,Q}(m):$} \\ \indent 
					$\Sigma' \gets Log(\Sigma,m)$ \\ \indent $//\Sigma'=[\Sigma'^{L},\Sigma'^{K},Cache']$ \\ \indent
					$Q \gets Q \cup  \Sigma'$\\ \indent
					Return $(\Sigma'^{L},cache'^L)$\\ \indent
				\end{tabular}
				\\
				\begin{tabular}{@{}l@{}}
					//\mbox{$REC$ runs $Recover(\cdot,\cdot)$, returns $\cc{R}$} \\ \indent
					\textbf{$REC_{\Sigma}():$} \\ \indent
					$\cc{R} \gets Recover(\Sigma,\Sigma_0)$ \\ \indent
					Return $\cc{R}$ \\ \indent 
				\end{tabular}
				&
				\begin{tabular}{@{}l@{}}
					//\mbox{$CRASH$ crashes the $\cc{L}$, returns $\Sigma^{cr}$} \\ \indent
					\textbf{$CRASH_{\Sigma}():$} \\ \indent
					Return $(\Sigma^{cr})$ 
				\end{tabular}
			\end{tabular}
			
		\end{algorithmic}
	\end{algorithm}
\end{scriptsize}
\vspace{-0.5cm}

	\section{An adaptive crash recovery scheme} \label{RewindSecure}
	\vspace{-0.1cm}
	
	The above impossibility  result shows that if no key information can be trusted after the crash, 
	it will not be possible to distinguish 
  between an accidental 
  crash and a crash attack. 
  One may use an external reliable storage such as 
	 blockchain 
	 \cite{SHEL,Catena}. In  such an approach the blockchain will 
	store 
	data that will allow the recovery algorithm to detect a crash  state. 
	Such a solution will  have challenges 
	 including the need for a high rate  of access to blockchain.
Our goal is to design a solution without using an external 
 point of trust. 
 	\remove{
	Although the initial state of the system is stored in a trusted location and is given to verifier later, this will not contradicts our goal. For any forward secure logging scheme that generates the keys from an initial seed, this assumption is needed.
We assume there is a protected component in the KStore  that will not be damaged by a normal crash.  {\color{red}This assumption can be implemented using protected non-volatile memory but it is not enough to ensure crash integrity.}
	The protection is only for the normal working of the system, and we assume that the memory will become accessible to the adversary after the compromise. {\color{red}If adversary is non-adaptive, we can implement the system without this component, so that log events and keys can be stored on the same disk.}
	  We design a logging protocol that prevents modification, deletion, and rewind attacks by using this assumption.}
 
	\subsection{The proposed scheme}
	We build the basis of our protocol close to the PRF-chain FI-MAC protocol of Bellare and Yee \cite{BY}. We assume that each log event is appended to the end of the log with an authentication tag, a HMAC. We use PRF to evolve the keys needed for HMAC. Multiple keys can be used in our scheme to prevent rewinding, but for simplicity, we describe  the mechanism with a pair of keys; the keys are used as below:
	
	\textbf{Double evolving key mechanism.} To prevent rewinding, we generate two key sequences that are evolved with different rates. One of the keys evolves per log entry to prevent re-ordering and log modification, and guarantees forward security. We call this key as sequential key. The second key, which is called state-controlled key, is updated slower relative to the first key at random points of time. This key is used to reduce the probability that key is removed from the disk after a normal crash. \remove{We give the rationale behind using this key later. Each state-controlled key lasts for one epoch, and whenever this key is updated we say a new epoch is started.}
	
	 For each log entry, we use a choice function $CF()$ which receives the index of the new log entry and the current state-controlled key $CF(k'_{j-1},i)$ and outputs 0 or 1. If the output is 0 we use the sequential evolving key and if it is 1 we use state-controlled key to compute the HMAC.

	 We require that state-controlled key evolves randomly, so attacker cannot guess or estimate the positions that KStore is updated.
	For this, we use a choice function $CF()$ which gets a random input and outputs 0 or 1. \remove{The reason for using a random output is to prevent the attacker from learning the update points. If the state-controlled key is updated in predetermined points then adversary can use this knowledge to rewind the log file to exactly that point where the success probability of deleting the key is high.} Thus, $CF()$ has the following properties:(i) by observing the input/output of $CF()$, adversary cannot predict the previous outputs; (ii) $CF()$ outputs $1$ with probability $\frac{1}{m}$.
	With this setting, we can say the state-controlled key is ``$\epsilon\_$stable'' \textit{relative to} the sequential evolving key.
	\\
	\textbf{Definition 3.} A key mechanism is called ``$\epsilon\_$stable'' if the probability that the key is removed by a normal crash is $\epsilon$.

	 We use $H(k'_{j-1},i) < T$ as our choice function $CF()$, where $H$ is a cryptographic hash function like SHA-256, $k'_{j-1}$ is the current state-controlled key, $idx$ is the index of the log entry that is going to be stored in LStore, and $T$ is a target value. $T$ is chosen such that the above equation holds with rate $\frac{1}{m}$ on average,  that is the state-controlled key is evolved with probability $\frac{1}{m}$ at each log entry. We show in Section \ref{Implementation}, how to determine $T$ for a given $m$ and prove the security of our scheme using this choice function in Theorem \ref{Theo}. 
	 A similar 
	 choice function  has been used 
	 in Bitcoin \cite{Bitcoin}. 
	 
	 Note that even by choosing a random choice function, adversary can find the index of the event corresponding to the last usage of state-controlled key. This can be done by exhaustive search in the tail end of the log file, using the HMAC on every event with the state-controlled key seen in the KStore. To also prevent this attack we require that the HMAC of the events where state-controlled key is updated should have a source of randomness. We use the state-controlled key before updating as this randomness and concatenate it with the event $m_i$, i.e. $h_i=HMAC_{k'_j}(m_i,k'_{j-1})$. Remember that KStore contains this key during the evolving process and removes it later on, so attacker cannot find it after compromise.
	 \remove{With this setting, If attacker compromises the system, it cannot find the previous state-controlled key to check the choice function on every event.}
	 It also worth to mention that adversary can only succeed in rewinding $\cc{L}$ to an old state if it forges the state-controlled key associated with that state. By using PRF to generate the key sequences this probability is negligible.
	 
	 \noindent \textbf{Details.}
	 Log file consists of a list of events $S=\{s_1, s_2, \dots\}$, where each element $s_i$ corresponds to one event. Each new event, $m_i$, is concatenated with a HMAC, $h_i$, and appended to $S=S||s_i$, and $s_i=(m_i,h_i)$, where $,$ denotes the concatenation and $||$ represents appending. The system algorithms are described in Algorithms \ref{Gen}, \ref{Log}, \ref{Recover}. 
	
	\vspace{-0.8cm}
	\begin{minipage}[!tb]{0.47\textwidth}
		\begin{minipage}{\textwidth}
			\begin{algorithm}[H]
				\scriptsize
				\caption{Gen$(1^{\lambda})$}\label{Gen}
				\hspace*{\algorithmicindent} \textbf{Input: Security parameter $\lambda$} \\
				\hspace*{\algorithmicindent} \textbf{Output: Initial state $\Sigma_0$}
				\begin{algorithmic}[1]
					\State $k_0 , k'_0 \gets \{0,1\}^\lambda$
					\State $\chi, \chi' \gets \{0,1\}^\lambda$
					\State Let $S \gets init\_message$ //S is a dynamic array
					\State \textbf{Output} $\Sigma_0=(\Sigma^{K}_{0},\Sigma^{L}_{0},Cache_0)$ \newline // where $\Sigma^{K}_{0}=(k_0,k'_0)$, $\Sigma^{L}_{0}=(S)$, and $Cache_0=\emptyset$;
				\end{algorithmic}
			\end{algorithm}
		\end{minipage}
		\vfill
		\vspace{1.2cm}
		\begin{minipage}{\textwidth}
			\begin{algorithm}[H]
				\scriptsize
				\caption{Log$(\Sigma_{i-1},m_i)$}\label{Log}
				\hspace*{\algorithmicindent} \textbf{Input: old state $\Sigma_{i-1}$, log event $m_i$} \\
				\hspace*{\algorithmicindent} \textbf{Output: updated state $\Sigma_i$}
				\newline \textit{ $\Sigma_{i-1}=[\Sigma^{K}_{i-1},\Sigma^{L}_{i-1},Cache_{i-1}]$, \newline $\Sigma^{K}_{i-1}=(k_{i-1},k'_{j-1})$ and $\Sigma^{L}_{i-1}=(S)$, $|S|=i-1$;} \newline \indent
				//new log event $m_i$ arrives
				\begin{algorithmic}[1]
					\State $k_{i}=PRF_{k_{i-1}}(\chi)$ \newline
					//Compute the choice function
					\If{$CF(k'_{j-1},i)=1$}
					\State $k'_{j}=PRF_{k'_{j-1}}(\chi')$
					\State $h_i=HMAC_{k'_j}(m_i,k'_{j-1})$
					\Else
					\State $h_i=HMAC_{k_i}(m_i)$
					\EndIf
					\State $s_i=(m_i,h_i)$
					\State $S=S||s_i$
					\State \textbf{Output} $\Sigma_i=[\Sigma^{K}_{i},\Sigma^{L}_{i},Cache_{i}]$ \newline //where $\Sigma^{K}_{i}=(k_{i},k'_{i})$ and $\Sigma^{L}_{i}=(S)$
				\end{algorithmic}
				\vspace{-1mm}
			\end{algorithm}
		\end{minipage}
	\end{minipage}
	\hfill
	\begin{minipage}{0.47\textwidth}
		\begin{algorithm}[H]
			\scriptsize
			\caption{Recover$(\Sigma,\Sigma_0)$}\label{Recover}
			\hspace*{\algorithmicindent} \textbf{Input: State $\Sigma$ to check, initial state $\Sigma_0$} \\
			\hspace*{\algorithmicindent} \textbf{Output: Recovered log events $\{(\rho_i,m'_i), 1 \leq i \leq n'\}$}\newline \textit{//Let $s_i=(m_i,h_i)$}
			\begin{algorithmic}[1]
				\State $\cc{R}=\emptyset$ (recover set), $ExpSet=\emptyset$ (expendable set) \newline
				//compute keys
				\For{$i=1 \quad to \quad n'+cs$}
				\State $k_{i}=PRF_{k_{i-1}}(\chi)$
				\State $\mathds{KS} \cup (i,k_i,\perp)$
				\If{$CF(k'_{j-1},i)=1$}
				\If{$i> n'$}
				\State $K' \cup k'_{j-1}$
				\EndIf
				\State $k'_{j}=PRF_{k'_{j-1}}(\chi')$
				\State $\mathds{KS} \cup (i,k'_j,k'_{j-1})$
				\State Remove  $(i,k_i,\perp)$ from $\mathds{KS}$
				\EndIf
				\EndFor 
				\State $K' \cup k'_{j}$\newline
				//verify HMACs using the key set $\mathds{KS}$ which is of form  $(i,k_i,\kappa_i)$ \newline where $\kappa_i$ is $\perp$ or $k'_{j-1}$
				\For{$i=1 \quad to \quad n'$}
				\If{$HMAC_{k_i}(m_i,\kappa_i)=h_i,  k_i, \kappa_i \in \mathds{KS}$}
				\State Update $\cc{R} \cup (i,m_i)$
				\EndIf
				\EndFor\newline //compute expendable logs
				\For{$i=n'-cs+1 \quad to \quad n'+cs$}
				\State $ExpSet \cup i$
				\EndFor\newline 
				//Plausibility check
				\If{$ X \notin K'$}
				\State Outputs $\perp$
				\ElsIf{$|R| < 1 \lor \exists i \in \{1, \dots, n'\}: \{(i,.) \notin \cc{R} \wedge i \notin ExpSet\}$}
				\State Outputs $\perp$
				\Else
				\State Outputs $\cc{R}$
				\EndIf
			\end{algorithmic}
		\end{algorithm}
	\end{minipage}

	\textbf{$Gen(1^{\lambda})$:} We use a PRF to generate and evolve the required keys. Let $PRF:\cc{K} \times \cc{Y} \rightarrow \cc{Z}$ be a function where $\cc{K}$ is the key space, $\cc{Y}$ is the domain and $\cc{Z}$ is the range, all are determined by security parameter $\lambda$. $PRF(k,\cdot)$ is often denoted by $PRF_k(\cdot)$.
	 There are two initial keys, one for computing sequential keys, denote it with $k_0$, and one for computing state-controlled keys, denote it with $k'_0$.
	  All the secrets are shared with the verifier at the beginning of the log file and they are removed from system after updating it to the next key. Note that PRF also takes a second input which does not need to be secret and it is stored at the logging device and also shared with  the verifier (We represent these inputs with $\chi$ and $\chi'$). PRF evolves as follows: $k_{i}=PRF_{k_{i-1}}(\chi)$ (similarly $k'_{i}=PRF_{k'_{i-1}}(\chi')$). State-controlled key is initially $k'_0$. $S$ is initialized with a specific message, which represent the information of log initialization such as the date, size, device id and etc; this is to prevent total deletion attack. We use $Log(.,.)$ algorithm that is described next to log the initial event, $init\_message$.  We assume that cache is initially empty, and the state of the $\cc{L}$ is  $\Sigma_0=(\Sigma^{K}_{0},\Sigma^{L}_{0},Cache_0)$, where the state of the KStore is  $\Sigma^{K}_{0}=(k_0,k'_0)$ and the state of the LStore is $\Sigma^{L}_{0}=(S)$.\\
	\textbf{$Log(\Sigma_{i-1},m_i)$:} Each log entry is of the form $s_i=(m_i,h_i)$ and it is appended to the dynamic array $S= S \cup (s_i)$, where $h_i$ is the HMAC of $m_i$ using either $k_i$ or $X=k'_{j}$.	For each log entry at index $i$, $CF(k'_{j-1},i)$ is calculated; if the output is 1 then $k'_{j-1}$ is updated to $k'_j$ and HMAC of $m_i$ is computed using $k'_{j}$ and $k'_{j-1}$, $h_i=HMAC_{k'_{j}}(m_i,k'_{j-1})$, otherwise $k_{i}$ is used for computing the HMAC, $h_i=HMAC_{k_i}(m_i)$. 
	Figure \ref{RewindPro} shows how Log algorithm works. When $CF()$ outputs 1, the corresponding log entry uses the state-controlled key. \\
	\textbf{$Recover(\Sigma,\Sigma_0)$:} Verifier $\cc{V}$ receives the state $\Sigma$ consisting of $n'$ log events (possibly crashed) in LStore. $\cc{V}$ knows the size of each log entry and can parse the LStore to $n'$ log entries. $\cc{V}$ also knows the initial state of the $\cc{L}$, so re-computes all the random coins and the keys and stores the keys in the set $\mathds{KS}$. $\cc{V}$ can verify the HMAC of each log entry using either the sequential key or the state-controlled key depending on the output of choice function $CF()$. The indexes between $n'-cs+1$ to $n'+cs$ are considered as expendable set. $\cc{V}$ also finds the set of possible  state-controlled 
	keys that  may be 
	in the KStore at the time of 
	 the crash.
	 After the  crash one such  key will be in  
	 KStore 
	  (lines 6-11). 
	  %
	If the size of the log file is $n'$,   
	  last key that has been updated before $n'$ will be in the KStore (because logging  is immediately after key update).
	  Since it is possible to have a situation where the  the cache contains a new event and  KStore contains the updated state-controlled key, but the corresponding event has not been written to the log file, we will have the following.
	  For a log file of length $n'$ the  key set $K'$ consists of   (i) state-controlled keys that are generated between index $n'$ and $n'+cs$ (future keys), and (ii) the last state-controlled key generated 
	  the event  $n'$.
 Figure \ref{SetK} shows how to find this key set. In this example, $cs=4$ and the size of log file $n'=11$, key $k'_3$ is associated with the last stored event, and $k'_4$ is associated with an unwritten event. So, $K'=\{k'_3,k'_4\}$. 
 
  \textit{Plausibility check.} If the state-controlled key, $X$, is not in the key set $K'$ then we output $\perp$ meaning untrusted log. If the number of recovered events are less than 1 there is a total deletion attack. If there is an index which is neither in the expendable set nor in the recovered set, then there is a deletion/modification attack. Otherwise, $Recover(.,.)$ outputs index-message pairs.

	\begin{figure}[!t]
		\begin{minipage}[t]{\textwidth}
		\centering
		\small
		\setlength{\belowcaptionskip}{-4mm}
		\includegraphics[width=0.85\textwidth]{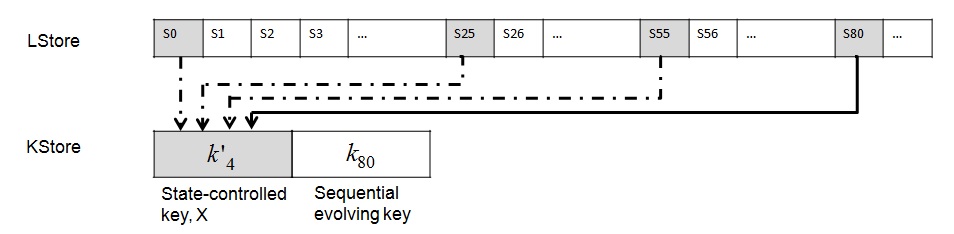}
		\caption{\small The logging operation using double evolving keys; log entry $s_{55}$ uses $k'_3$ and then $k'_3$ evolves to $k'_4$ which is used for $s_{80}$.} \label{RewindPro}
		\vspace{-0.5mm}
			\end{minipage}
	\end{figure}
\vfill
	\begin{figure}[!t]
		\begin{minipage}[t]{\textwidth}
	\centering
	\setlength{\belowcaptionskip}{-4ex}
	\includegraphics[width=0.75\textwidth]{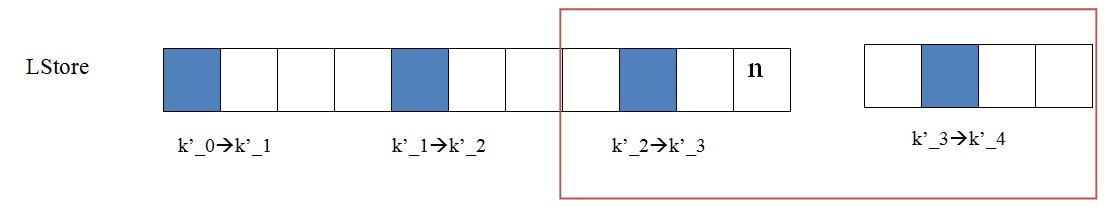}
	\caption{\small Estimating the set of possible keys. Blue cells are the locations that state-controlled key has been updated and the red rectangular shows the Expset.} \label{SetK}
\end{minipage}
\end{figure}

 
 	 \noindent

	\section{Security and efficiency} \label{SecurityAnalysis}
First we give two lemmas;  Lemma \ref{Key} shows the stability of our mechanism and Lemma \ref{lemset} shows the relation between $\frac{1}{m}$, and the number of log entries that adversary can truncate from the end of log, $\ell$. Then, we give Theorem \ref{Theo}, for which we follow the \textit{$Exp^{AdapCr}_{\cc{A}_{a},\Pi}()$} game and assume that after compromise, attacker has access to everything in the system including the choice function, the keys in the KStore and the key cache. (Please see the details of proofs in Appendix \ref{Proofs}.)

\begin{lemma} \label{Key}
	The double evolving key mechanism is $\frac{\alpha^2}{m}\_$stable if the choice function $CF()$ outputs 1 with probability $\frac{1}{m}$.
\end{lemma}

\remove{
\begin{proof}
	The proof is straightforward. Note that key can be removed from KStore if key is being updated at the moment that crash happens. (Recall that in the update procedure of key, (i) the current key $k_{i-1}$ goes to cache and updated to a new key $k_i$, (ii) $k_i$ is written to KStore and (iii) $k_{i-1}$ is deleted. Assuming that the logging system allows reordering of this process, deletion of $k_{i-1}$ can happen before $k_i$ is written on KStore. A crash at this moment leads to a state where the keys are missing from both KStore and cache). To prove the lemma, we should compute $P(\mbox{deletion} \wedge \mbox{key updates})$. Let $\alpha$ denotes the probability that a key is deleted from KStore (because of re-ordering procedure in the system). If state-controlled key evolves with probability $\frac{1}{m}$ at each event, the probability that key is removed during a normal crash will be  $\alpha \times \frac{1}{m}$.
\end{proof}
}

 Let $\alpha$ denotes the probability that a key is deleted from KStore (because of re-ordering procedure in the system). If state-controlled key evolves with probability $\frac{1}{m}$ at each event, the probability that both 
 sequential 
 and state-controlled keys, are removed during a normal crash will be  $\alpha \times \frac{\alpha}{m}$.
 Note that by choosing large values for $m$ this probability becomes negligible.
   If we use two  
   independent state-controlled keys using 
  different PRFs, and evolving 
  at different rates 
(
$\frac{1}{m_1}$ and 
$\frac{1}{m_2}$, respectively),   then the probability that  after a normal crash, the sequential key and both state-controlled keys are missing 
 will be reduced to $\alpha \times \frac{\alpha}{m_1}\times \frac{\alpha}{m_2}$. This method can be used to dramatically decrease the chance of key removal key in a normal crash if we do not want to increase the value of $m$ directly. 

Note that we cannot unlimitedly increase $m$. If $m$ is chosen to be so large, attacker may want to keep the state-controlled key untouched and truncate the log file for a number of events, $cs+\ell$, such that the key is also valid in the truncated state (note that we are interested in the value of $\ell$, since we do not have any guarantee for $cs$ events in a normal crash anyway). Consider that attacker compromises the logging device at state $n$ where the set of possible keys is $K'_n$ and then cuts off $cs+\ell$ log events from the end which results in the malicious state $n'$; the set of possible state-controlled keys is denoted by $K'_{n'}$ at state $n'$. If $K'_n$ and $K'_{n'}$ have intersection and the key in the KStore is one of the keys in the intersection of the two key sets then verifier cannot distinguish the crash attack from the accidental crash and hence crash attack ends up successfully. The value of $\ell$ is important to the security of our system and our goal is to reduce $\ell$. In the following lemma \ref{lemset}, we find the success probability of attacker in the attack mentioned above; attacker knows the key evolving probability $\frac{1}{m}$ and the size of the log file.

\begin{lemma}\label{lemset}
	Assuming that the evolving probability of state-controlled key is  $\frac{1}{m}$, and attacker compromises the device at sate $\Sigma_n$, truncates $cs+\ell$ events from the log file, results in malicious state $\Sigma_{n'}$, and keeps the key in KStore untouched. The success probability of such attacker is bounded to $P_s=(1-\frac{1}{m})^\ell\times\frac{1}{\lfloor\frac{cs}{m} \rfloor+1}$. 
\end{lemma}

\begin{theorem}\label{Theo}
	 Our construction provides $[\epsilon_{PRF}(\lambda),\epsilon_{PRF}(\lambda),f(n,n',\ell,cs,\lambda)]$-Crash Integrity against an adaptive attacker $\cc{A}$, where PRF-HMAC is $\epsilon_{PRF}(\lambda)\_$secure, $\ell$ is the number of events adversary wants to delete, $cs$ is the cache size, $n$ is the size of log file at state $\Sigma_n$, $n'$ is the number of log entries returned by adversary in the malicious state $\Sigma_{n'}$, $\lambda$ is the security parameter, and $f()$ is as follows:
	
	\begin{minipage}{\textwidth}
	\mathleft
	\begin{equation}
	\footnotesize
	f=\begin{cases}
	0, & \text{if $n'<1$}\\ max\{\epsilon_{PRF}(\lambda),(1-\frac{1}{m})^\ell\times \frac{1}{\lfloor\frac{cs}{m}\rfloor+1}\}, & \text{otherwise}\\
	\end{cases}
	\end{equation}
	\end{minipage}
	
\end{theorem}
\vspace{-0.5cm}
The  theorem shows that 
$m$ can be chosen to make 
 the success probability of truncating the log  becomes negligible. This choice however will result in small value of $m$ (bigger than 1).
 %
 Achieving 
 $\epsilon\_$stability requires 
  large values of $m$,  while   crash integrity suggests 
  small $m$.
  By using 
   multiple evolving keys we can keep 
    $m$ small while achieving 
     $\epsilon\_$stability guarantee of the mechanism.
     This is because each key has a small evolving probability, so the probability that all keys are removed at the same time will be negligible. 
     If  attacker  truncates the file by more  than 
     $cs$ events,  there is at least one key that will be affected and this will 
     reveal the attack.
     The number of keys, $n_{sc}$,  will depend on the probabilities $\{\frac{1}{m_1}, \frac{1}{m_2},...,\frac{1}{m_{n_{sc}}}\}$.
     One can also choose 
     a different distribution for the choice function. By using uniform distribution for double evolving key mechanism, adversary can
     truncate the file by 
     at most $m$ log entries, 
     with success probability $\frac{m-\ell}{m+cs}$ for $\ell<m$. 
     Appendix \ref{Uniform} gives details of this analysis. 
     Finding the best probability distribution for $CF()$ to 
      minimizes the success probability of the attacker  is an interesting future research direction. 


\vspace{-0.3cm}
\subsection{Complexity analysis}\label{Complexity}
\vspace{-0.3cm}
According to $Log(\cdot,\cdot)$ defined in algorithm \ref{Log}, the complexity of adding one event is $O(1)$ since it needs  (i) evolving the keys and (ii) computing the HMAC, and hence there are constant number of disk operations. Although in SLiC \cite{CrashAttack} the computational complexity  of  logging is 
 $O(1)$, our proposed system is faster: the required computation in SLiC consists of (i) updating the keys,(ii) encrypting the log event, and (iii) performing a local permutation on the log file.  Additionally, each log operation in our scheme requires one write operation on disk whereas in SLiC each log operation requires two write operations. Moreover, in our system the order of events is preserved in the log file, so that searching a specific event is efficient. The complexity of $Recover(\cdot,\cdot)$ in our scheme (algorithm \ref{Recover}) for verifying the total number of $n'$ events is equal to $O(n')$; the first and the second loop in algorithm \ref{Recover} takes $O(n')$ computations, the third loop has complexity of $O(1)$ and the plausibility check has $O(n')$. In SLiC \cite{CrashAttack}, the complexity of recover algorithm is $O(n'log(n'))$ since it needs to run sort algorithm for verification. The complexity of our scheme is less than SLiC, but it is the same as  SLiC$^{Opt}$ \cite{CrashAttack} (please see Table \ref{Complexity}).
\vspace{-0.8cm}

\begin{table}[h!]
	\footnotesize
	\centering
	\caption{Comparison between computation complexity of our scheme and SLiC} \label{Complexity}
	\begin{tabular}{|c |c |c |c|}
		\hline
		Algorithm & Our scheme & SLiC \cite{CrashAttack}& SLiC$^{Opt}$ \cite{CrashAttack}\\
		\hline
		$Log(\cdot,\cdot)$ & $O(1)$ & $O(1)$ & $O(1)$ \\
		\hline
		$Recover(\cdot,\cdot)$ & $O(n')$ & $O(n'log(n'))$ & $O(n')$\\
		\hline
		\end{tabular}
\end{table}
\vspace{-0.9cm}
\section{Implementation and evaluation}\label{Implementation}
\vspace{-0.3cm}
We implement and evaluate the double evolving key mechanism in Python. The experiments are run on two hardware platforms:  a windows computer with 3.6 GHz Intel(R) Core(TM) i7-7700 CPU; a Raspberry Pi 3, Model B with 600 MHz ARM CPU running Raspbian. \\
\noindent
\textbf{Logging Performance:} We measure the logging performance on a prepared text file as the source of system events. The text file contains $2^{20}$ random strings, each with 160 characters. To implement our log scheme, we use ChaCha20\cite{Chacha20} for PRFs and  SHA-256 as hash function in HMACs. We find  the cache size of our machine using the same approach explained in \cite{CrashAttack}; the maximum UDP packet sending rate is 500 event/s (please see Appendix \ref{AppCache} for the result of our experiment).  Accordingly, the cache size, $cs$ is $15000\approx 2^{14}$ events considering the page eviction time of $30s$.  We set $m=cs$ and $T$ value in the $CF()$ is determined to be $2^{242}={2^{256}}/{2^{14}}$ which outputs 1 with probability ${1}/{2^{14}}$ \cite{DifMining} (please see Appendix \ref{AppCache} for the related experiments). The length of both keys is 256 bits. We implement two logging schemes to compare with our logging scheme: (i) \textit{Plain scheme:}  Each system events is stored in the log file as plaintext. (ii) \textit{SLiC:} The logging algorithm proposed in \cite{CrashAttack}.  We initialize with $\lambda=2^{15}$ randomly ordered dummy events as the same in \cite{CrashAttack}. We implement PRFs using ChaCha20 \cite{Chacha20},  HMACs using SHA-256 and encryption functions using AES-CTR-256. The size of the key in is 256 bits. 

By comparing our scheme with the plain scheme we can find the extra cost to provide crash integrity.  
We also compare our scheme with SLiC to find the extra cost of protecting against an adaptive adversary. Table \ref{Eval1} shows the total runtime to log $2^{20}$ system events using three aforementioned logging schemes on Windows PC and Raspberry Pi. We repeat the experiments for 5 times, each time with a new file containing $2^{20}$ events (the other settings remain same).
\vspace{-0.8cm}
\begin{table}[h!]
	\footnotesize
	\begin{center}
		\caption{\small The total time (in seconds) to log $2^{20}$ events}\label{Eval1}
		\begin{tabular}{ |c|c|p{1.3cm}|p{1.3cm}|p{1.3cm}|p{1.3cm}|p{1.3cm}| }
			\hline
			Hardware & Scheme & Exp1 & Exp2 & Exp3 & Exp4& Exp5 \\
			\hline
			\multirow{3}{6em}{Windows PC} & Our Scheme & 40.2 &40.2&40.4&40.7&40.5\\ 
			& SLiC & 95.2 &96.0&95.2&95.4&96.0\\ 
			& Plain & 2.0 &2.0&2.0&2.0&2.0\\ 
			\hline\hline
			\multirow{3}{6em}{Raspberry Pi} & Our Scheme & 330.5 &325.4&319.0&324.5&319.6\\ 
			& SLiC & 790.2&792.0&777.9&789.2&796.8\\ 
			& Plain & 18.8 &18.7&18.8&19.0&18.9\\ 
			\hline
		\end{tabular}
	\end{center}
	\vspace{-0.8cm}
\end{table}
For the same hardware and the same logging schemes, but with different files, the runtime remains same. This is aligned with our expectation: logging performance is independent of file content. On the windows PC, our scheme takes $\approx 40$ s ($\approx 26K$ events/s) on average. This represents a multiplicative overhead of 20, compared to the plain scheme, while SLiC takes $\approx 95$s, with an overhead of 47.  Compared to our log scheme, SLiC has a multiplicative overhead of 2, while our scheme provides extra security protection. However, in \cite{CrashAttack}, they observed a slowdown factor of 20 for logging rate on a laptop. The PRF they chose or the difference between their hardware and ours may cause the discrepancy of the result. Unfortunately, We could not find any detailed information regarding the implementation of PRF in \cite{CrashAttack}. 

The runtime on the Pi is roughly 8 times the runtime on the desktop, because of the computational limit of Raspberry Pi. The results shows that our logging scheme is still lightweight for the resource-constrained device.  It takes $\approx 324$ s ($\approx 3.2K$ events/s) on average to log $2^{20}$ events. The overhead of our scheme compared to plain scheme is 17.  SLiC has a multiplicative overhead of 2 compared to our scheme, and an overhead of 42 compared to plain scheme.\\
\noindent
\textbf{Recovery Performance:} Normally, we assume the logging results are written to a file in the OS. If crash happens, the verifier can always get the number of events ($n^{'}$) based on the size of the file. In our implementation, the value of $n^{'}$ is $2^{20}$. We run our $Recover(.,.)$ algorithm on the five log files generated earlier before. Table \ref{Eval} shows the total runtime to recover log files on two platforms. It takes $\approx 37.4$s on average to recover all the system events on the desktop and $\approx 308.4$s on the Pi. We observe that it takes slightly more time for logging than recovery, maybe because of the poor I/O handling of Python. In our implementation of log algorithm, the key is evolved per new line from the I/O. While in the implementation of recovery algorithm, all of the keys are reconstructed before any reading from the I/O.

\begin{table}[!t]
	\begin{center}
		\caption{\small  The total runtime (in seconds) to recover a log file of size $2^{20}$ events}\label{Eval}
		\begin{tabular}{|c| c |c| c |c |c|}
			\hline
			Hardware & Exp 1 & Exp 2 & Exp 3&Exp 4&Exp 5\\ [0.5ex] 
			\hline\hline
			Windows PC & 37.1 & 37.6 & 37.5 &37.3&37.3  \\ 
			\hline
			Raspberry Pi & 311.0 & 303.1& 302.0&303.4& 303.3  \\
			\hline
		\end{tabular}
	\end{center}
	\vspace{-0.8cm}
\end{table} 

\vspace{-0.4cm}
\section{Conclusion}\label{Conclusion}
\vspace{-0.3cm}
We proposed  adaptive crash attack where adversary can see intermediate states of the logging operation. By compromising the logging device, adversary can rewind the system back to one of the past states and then crash it to appear as a normal crash. We showed that this attack is strictly stronger than non-adaptive crash attack and all existing schemes are subjective to this attack. We also proposed double evolving key mechanism as a protection against rewinding which basically relies on two sequences of keys evolving with different rates. The security of scheme is proved and the performance of our approach is evaluated on both a desktop and Raspberry Pi.  Ensuring crash integrity against an adaptive attacker without considering a protected memory for keys is left as future work.

\noindent{\bf Acknowledgments.}
This work is in part supported by a research grant from Alberta Innovates in the Province of Alberta in Canada.
\vspace{-0.3cm}
	
	%
	%
\bibliographystyle{splncs04}
\bibliography{sample-bibliography}

\appendix

	\section{Detail of SLiC protocol \cite{CrashAttack}}
	\label{SLiCDiscription}
	\vspace{-0.3cm}
	There are three algorithms, $Ge(\cdot)$, $Log(\cdot,\cdot)$, and $Recover(\cdot,\cdot)$, as follows:
	
	\textbf{$Gen(1^\lambda)$:} An initial key $K_0$ and seed $seed_0$ are chosen uniformly from random. Log
	entries are stored in a dynamic array $S$. $S$ is initialized by storing $\lambda$ dummy events in the random order. To add these dummy elements
	$dummy_1;...;dummy_{\lambda}$ to $S$, $Log$ mechanism is used. The output
	of $Gen$ is the initial state $\Sigma_0$ comprising the key $K_0$, seed $seed_0$, and array $S$. State $\Sigma_0$ is shared between $L$ and $V$.

	\textbf{$Log(\Sigma_{i-1},m_i)$:} For log event $m_i$, the log entry $s_i$ is computed as $s_i = (c_i; h_i;\kappa_i)$ and appended to $S$; where $s_i$ is a dynamic array, $c_i$ is the encryption of event $m_i$ using key $K_i$ ($c_i = Enc_{K_i}(m_i)$) and $h_i$ is the HMAC of $m_i$ using the same key ($h_i = HMAC_{K_i}(c_i)$).
	The key $K_i$ is updated through a PRF $K_{i+1} = PRF_{K_i}(\chi)$
	for some constant $\chi$. The sorting key $\kappa_i$ is computed using a PRF ($\kappa_i=PRF_{K_{i-1}}(i)$) where $i$ is the index of the event.
	Then by using a PRG, $PRG(seed_i)$,  a position $pos$ is selected between $1$ and $\lambda+i$. $seed_i$ is updated to $seed_{i+1}$ through a PRF, $seed_{i+1} =PRF_{seed_i}(\chi')$.
	Position $pos$ is the
	position where the new log entry $s_i = (c_i; h_i; \kappa_i)$ is stored. If 	$pos \neq \lambda + i$, the log entries are swapped; $S[pos]$ is appended to $S$ and $s_i$ is written at
	$S[pos]$. Finally,  $K_{i-1}$ is evolved to $K_i$ and $seed_{i-1}$ to $seed_i$.
	
	\textbf{$Recover(\Sigma,\Sigma_0)$:}
	On receiving state $\Sigma$ as input, which comprises of log entries $\pi'=(s'_1,...,s'_{n'})$, the log file is parsed to $n'$
	log entries (which may be broken).
	$K_0$ and $seed_0$ are moved forward to the earliest possible time
	of a valid state. Expendable log
	events $\varepsilon$ are considered from indexes $n'-cs$ till $n'+cs$. To compute $\varepsilon$, the permutation $\pi$ is reconstructed to
	determine where log entries should be located in a valid
	$S$. Here, $\pi[i] = j$ denotes that log entry $s_i$ resides at $S[j]$.
	Then, $V$ replays
	$L$'s random coins and the swaps.
	The log entries $s'_i$ are sorted
	by their keys $\kappa_i$
	and stored in a
	binary search tree. $V$ then checks the HMAc of all $n'+cs$ log entries. The events that their HMAC is valid will be added to set $\cc{R}$. Finally, $V$
	perform a plausibility check to distinguish a regular crash from
	a crash attack. If the size of recovered set $\cc{R}$ is less than $\lambda-cs$, then there is a complete deletion. If log event $m_i$,  $1 \leq n'-cs$ is not recovered, it must be in the set of
	expendable events $\varepsilon$. If not, there is a deletion attack. Otherwise, $V$ outputs the recovered set $\cc{R}$.
	
\remove{
	\section{Timeline of crash attack} \label{AppT}

Figure \ref{timeline2} shows the timeline of adaptive crash attack versus non-adaptive crash attack.

\begin{figure}[h]
	\centering
	\setlength{\belowcaptionskip}{-2ex}
	\includegraphics[width=0.65\columnwidth]{Timeline2.jpg}
	\caption{Timeline of non-adaptive crash attack \cite{CrashAttack} versus our adaptive model.} \label{timeline2}
\end{figure}
}

\section{Impossibility result}
\label{Impossibility}
\textbf{Lemma A.}
	If there exists a non-adaptive adversary $\cc{A}_{na}$ who plays the \\ \textit{$Exp^{CrInt}_{\cc{A}_{na},\Pi,Crash}()$} game \cite{CrashAttack} and successfully deletes or modifies at least one event with probability $\epsilon$ from $\cc{L}$, then there exists an adaptive adversary $\cc{A}_{a}$ who interacts with $\cc{A}_{na}$ and wins the  \textit{$Exp^{AdapCr}_{\cc{A}_{a},\Pi}()$} game with probability \textit{at least} $\epsilon$. 

\begin{proof}
	To prove the lemma, we assume that $\cc{A}_{a}$ responds to Gen, Log, and Crash queries made by $\cc{A}_{na}$ with the help of $GEN_Q()$, $LOG_{\Sigma,Q}()$ and $CRASH_{\Sigma}()$ oracles in the \textit{$Exp^{AdapCr}_{\cc{A}_{a},\Pi}()$} game and the challenger $\cc{C}$. The goal of $\cc{A}_{a}$ is to win the \textit{$Exp^{AdapCr}_{\cc{A}_{a},\Pi}()$} game using the information outputted by $\cc{A}_{na}$.
	
	First, challenger $\cc{C}$ runs $Gen(1^{\lambda})$ and initializes LStore, KStore and cache (line 2 of \textit{$Exp^{CrInt}_{\cc{A}_{na},\Pi,Crash}()$} and \textit{$Exp^{AdapCr}_{\cc{A}_{a},\Pi}()$} ). Then,  $\cc{A}_{na}$ sends $n$ messages $m_i$, $1 \leq i \leq n$ to $\cc{A}_{a}$ to log (lines 3 to 5 of \textit{$Exp^{CrInt}_{\cc{A}_{na},\Pi,Crash}()$}). On receiving the message $m_i$, $\cc{A}_a$ adaptively calls the $LOG_{\Sigma,Q}()$ oracle, and receives the state of the LStore $\Sigma^{L}_i$ and its cache $cache^{L}_i$. Note that $\cc{A}_{na}$ can also call Gen and Crash oracles which are required to be handled by $\cc{A}_a$. In case of Gen queries, $\cc{A}_a$ calls $GEN_Q()$ oracle which initializes a new log on $\cc{L}$ and for Crash queries, $\cc{A}$ calls $CRASH_{\Sigma}()$. $\cc{A}_{na}$ removes or modifies some events, and returns the crashed state $\Sigma^{cr}$ and the positions that it has modified or deleted $\alpha_1, \cdots, \alpha_{\ell}$ (\textit{$Exp^{CrInt}_{\cc{A},\Pi,Crash}()$} line 6)). $\cc{A}_{a}$ outputs whatever $\cc{A}_{na}$ outputs and wins the game with the probability \textit{at least}\footnote{The adaptive adversary $\cc{A}_{a}$ has also seen intermediate states whereas $\cc{A}_{na}$ succeeds without observing those states. In the worst case, the extra knowledge that $\cc{A}_{a}$ has, does not help it and the success probability is $\epsilon$.} equal to $\epsilon$.
	
\end{proof}

\textbf{Lemma B.}
	if there is a scheme $\zeta$ which provides crash integrity (according to \\ \textit{$Exp^{CrInt}_{\cc{A}_{na},\Pi,Crash}()$} game),  it can be broken by an adaptive adversary $\cc{A}_a$ who plays \textit{$Exp^{AdapCr}_{\cc{A}_{a},\Pi}()$} game.

\begin{proof}
	We prove the lemma using SLIC \cite{CrashAttack} as $\zeta$. \remove{SLIC \cite{CrashAttack} uses a permutation on the log file; whenever a new log event is received, it  is swapped with an older event chosen randomly from the past log entries.} We follow the \textit{$Exp^{AdapCr}_{\cc{A}_{a},\Pi}()$} game on SLIC \cite{CrashAttack}. 
	First, challenger $\cc{C}$ calls $Gen(1^\lambda)$ algorithm in \cite{CrashAttack} which initializes, (i) the LStore using $Log(\cdot,\cdot)$, (ii) the KStore to hold the required keys, and (iii) the cache to be empty. It returns the LStore $\Sigma^{L}_0$ and its associated cache $cache^{L}_0=\emptyset$ to $\cc{A}_a$. $\cc{A}_a$ calls $LOG_{\Sigma,Q}()$ to log $n$ messages $m_i$, where $1 \leq i \leq n$. Each message is logged using $Log(\cdot,\cdot)$ algorithm in \cite{CrashAttack} and the LStore $\Sigma^{L}_i$ and $cache^{L}_i$ are returned to $\cc{A}_a$.  Finally, $\cc{A}_a$ compromises $\cc{L}$, rewinds its state to be as $\Sigma'_{n'}=[\Sigma^{L}_{n'},\emptyset,(cache^{L}_{n'},\emptyset)]$, where $1 \leq n' < n-cs$ and calls $CRASH_{\Sigma}()$. Then $\cc{A}$ returns the crashed state $\Sigma^{'cr}_{n'}$ and all the positions $\alpha_i$ it has removed where $\alpha_i \in \{n'+1,\cdots,n\}$ (note that $n'+cs+1,..,n$ are not in the expendable set of $\Sigma^{'cr}_{n'}$). Challenger $\cc{C}$ runs $Recover(\cdot,\cdot)$ of \cite{CrashAttack} to get $\c{R}$. Since $\Sigma'_{n'}$ is a valid state $\cc{R}$ is not $\perp$, but the $\alpha_i$s are missing from it. Therefore, adversary wins the game and successfully deletes the events from $\cc{L}$ without being detected.
	
\end{proof}	
\vspace{-0.3cm}
\section{Proof of Lemma \ref{Key} and Lemma \ref{lemset}} \label{Proofs}
\textbf{Proof of Lemma \ref{Key}.}
	To prove the lemma, we should compute $P(\mbox{deletion} \wedge \mbox{key updates})$. Note that key can be removed from KStore if key is being updated at the moment that crash happens. (Recall that in the update procedure of key, (i) the current key $k_{i-1}$ goes to cache and updated to a new key $k_i$, (ii) $k_i$ is written to KStore and (iii) $k_{i-1}$ is deleted. Assuming that the logging system allows reordering of this process, deletion of $k_{i-1}$ can happen before $k_i$ is written on KStore. A crash at this moment leads to a state where the keys are missing from both KStore and cache).  Let $\alpha$ denotes the probability that a key is deleted from KStore (because of re-ordering procedure in the system). If state-controlled key evolves with probability $\frac{1}{m}$ at each event, the probability that key is removed during a normal crash will be  $\alpha \times \frac{1}{m}$.

\textbf{Proof of Lemma \ref{lemset}.}
	Assume that the set $K'_n$ represents the possible state-controlled keys after crash in state $\Sigma_n$, which consists of the last generated key before index $n$ and all keys from $n$ to $n+cs$. The size of $K'_n$ is on average equal to $\lfloor\frac{cs}{m}\rfloor+1$ based on the frequency of key evolving.
	When adversary compromises the device, one of the keys in $K'_n$ resides in KStore; each key has probability $\frac{1}{|K'_n|}$ to be in KStore. Assume that adversary wants to keep that key, and truncate $cs+\ell$ entries from the log. 
	 If the key which is in the KStore is also in the possible key set at state $n'$, $K'_{n'}$, attacker succeeds. To find the success probability of attacker we need to first consider all the cases that $K'_n$ intersects with $K'_{n'}$, and calculate their probability. (Remember that the new log file has length of $n'=n-cs-\ell$, and the ExpSet for the malicious state $\Sigma_{n'}$ is between $n-2cs-\ell+1$ to $n-\ell$.)
	
	\noindent There are two possible cases, in each of them there is a possibility that at least one state-controlled key is in both possible key sets of $K'_n$ and $K'_{n'}$.
	\begin{itemize}
		\item Case 1: The last key evolved before $n$ is also the last key evolved before $n'$, so both key sets have the same last key. 
		\item Case 2: The last key evolved before $n$ is not the same as the last key evolved before $n'$, but it evolves between $n'+1$ to $n'+cs$, so it is also in $K'_{n'}$ .
	\end{itemize}  
	
	Case 1 implies that there is no evolving between $n'$ to $n$. The length of this interval is $cs+\ell$. This probability equals to (binomial distribution) $P_1=(1-\frac{1}{m})^{cs+\ell}$.
	
	Case 2 implies that we have at least 1 key evolving between $n'+1$ and $n'+cs$ (the length of interval is $cs$) and there is no key evolving between $n'+cs+1$ till $n$ (the length of interval is $\ell$); This leads to the following probability:
	
	\noindent
	$P_2=[1-(1-\frac{1}{m})^{cs}](1-\frac{1}{m})^{\ell}=(1-\frac{1}{m})^{\ell}-(1-\frac{1}{m})^{cs+\ell}$
	
	Attacker succeeds if either case 1 or case 2 happens and the key in the KStore is the one which is in the intersection of the two key sets. So, the success probability of attacker is bounded to: $P_s=(P_1+P_2)\times \frac{1}{|K'_n|}$
	
	\noindent	
	By substituting $P_1$ and $P_2$ and by considering $|K'_n|= \frac{cs}{m}+1$, we get:
	
	\noindent	
	$P_s=(1-\frac{1}{m})^\ell\times \frac{1}{\lfloor\frac{cs}{m}\rfloor+1}$

\textbf{Proof of Theorem \ref{Theo}.}
	Let assume that $\epsilon_{PRF}(\lambda)=max\{\epsilon_{PRF_1}(\lambda), \epsilon_{PRF_2}(\lambda)\}$, where $\epsilon_{PRF_1}(\lambda)$ is the success probability of attacker in breaking PRF-HMAC \cite{BY} for the sequential key and $\epsilon_{PRF_2}(\lambda)$ is the success probability of attacker in breaking PRF-HMAC for the state-controlled key.\\
	The success probability of attacker in modification is equal to success probability of breaking PRF-HMAC (either the sequential key or the state-controlled key), so it is $\epsilon_{PRF}(\lambda)$.\\
	For deletion and rewinding attacks, if adversary deletes 1 log event from the log file at position $pos<n-cs$ it should forge the PRF-HMAC for all log entries greater equal to $pos$.
	so, the integrity of log file is reduced to the security of PRF-HMAC (either the sequential key or the state-controlled key).\\
	If adversary does a total deletion attack, $n'<1$, the attack is detected and hence the success probability is $0$. If attacker rewinds (or truncates) the system back to $n' = n-(cs+\ell)$, in which $cs+\ell$ events has been truncated and keeps the key in the KStore without change, the success probability of attacker will be $max\{\epsilon_{PRF_1}(\lambda),P_s\}$, based on the Lemma \ref{lemset}. If $m$ is chosen such that for $\ell>\ell_0$ the attack mentioned in Lemma \ref{lemset} becomes negligible, then attacker should modify the key in the KStore for $\ell>\ell_0$; success probability of such attacker is bounded by the security of HMAC-PRF, which is $\epsilon_{PRF_2}(\lambda)$.

\vspace{-0.4cm}
\section{Uniform choice function}\label{Uniform}
\begin{lemma}\label{lemset2}
	Assuming that the evolving rate of state-controlled key is  $\frac{1}{m}$ over each interval size of $m$ where $m <cs$, and attacker compromises the device at sate $\Sigma_n$, truncates $cs+\ell$ events from the log file, results in malicious state $\Sigma_{n'}$, and keeps the key in KStore untouched. The success probability of such attacker is bounded to $P_s=\frac{m-\ell}{m+cs}$ for $\ell<m$. 
\end{lemma}

\begin{proof}
The proof is similar to Lemma \ref{lemset} and we do not repeat the common parts for brevity. We need to only re-compute the probability of two possible cases stated in Lemma \ref{lemset} that lead to an intersection between key sets:

Case 1 implies that there is no evolving between $n'$ to $n$. The length of this interval is $cs+\ell$ and $cs+\ell > cs> m $; based on the evolving rate we know that there is one key evolving in this interval, so $p_1=0$ and case 1 is not probable. 

Case 2 implies that we have 1 key evolving between $n'+1$ and $n'+cs$ (the length of interval is $cs$) and there is no key evolving between $n'+cs+1$ till $n$ (the length of interval is $\ell$); as $cs>m$ the first statement is always true and there is 1 key evolving between $n'+1$ and $n'+cs$. The second statement, is not probable if $\ell > m$, since we have 1 key evolving, but if $\ell <m$ then with probability $P_2=1-\frac{\ell}{m}$ we have no evolving. Therefore, attacker succeeds if case 2 happens and the key in the KStore is the one which is in the intersection of the two key sets. So, the success probability of attacker is bounded to: $P_s=P_2\times \frac{1}{|K'_n|}, \ell<m$.

\noindent	
By considering $|K'_n|= \frac{cs}{m}+1$, we get:	
$P=(1-\frac{\ell}{m})\times \frac{1}{\frac{cs}{m}+1}=\frac{m-\ell}{m+cs},\ell<m$.
\end{proof}

\vspace{-0.5cm}
\remove{Note that adversary can truncate at most $cs+m$ log entries. If we choose $m$ to be equal to $\frac{cs}{3}$, then adversary can truncate maximum of $cs+\frac{cs}{3}$, and the probability for $\ell<\frac{cs}{3}$ is $\frac{3cs-9\ell}{4cs}$. }
\begin{theorem}\label{Theo2}
Our construction, using a uniformly distributed choice function over $m$, provides $[\epsilon_{PRF}(\lambda),\epsilon_{PRF}(\lambda),f(n,n',\ell,cs,\lambda)]$-Crash Integrity against an adaptive attacker $\cc{A}$, where PRF-HMAC is $\epsilon_{PRF}(\lambda)\_$secure, $\ell$ is the number of events adversary wants to delete, $cs$ is the cache size, $n$ is the size of log file at state $\Sigma_n$, $n'$ is the number of log entries output by adversary in the malicious state $\Sigma_{n'}$, $\lambda$ is the security parameter, and $f()$ is as follows:

\begin{minipage}{\textwidth}
\mathleft
\begin{equation}
\footnotesize
f=\begin{cases}
0, & \text{if $n'<1$}\\ \frac{m-\ell}{m+cs}, & \text{else if $n'\geq n-(cs+\ell)$}\\ \epsilon_{PRF}(\lambda), & \text{else if $n'<n-(cs+\ell)$}
\end{cases}
\end{equation}
\end{minipage}
\end{theorem}

 \section{Extra experiments}\label{AppCache}
 \noindent
 \textbf{Estimation of Cache size} We used the same method described in \cite{CrashAttack} to estimate the Cache size of our Windows computer. We used a remote computer as the syslog server to send the system events to our windows computer with different rates. Both computers are in the same LAN and the communication protocol is UDP. To calculate the packet drop rate, we divided the length of  received data on our computer by the length of the sent data from the syslog server. The results are shown in Figure \ref{drop}. The drop rate starts to grow when the sending rate increases to 500 events/s.
 
 \noindent
 \textbf{State-controlled key update} To validate whether the update frequency of the state-controlled key in practical is close to the theoretical value of ${1}/{2^{14}}$, we record the distance (number of system events) between every two state-controlled keys and collected the data out of  five experiments on Windows PC. The data is plotted to a histogram in Figure \ref{distribution}. The mean value of the distance is close to the desired value, which is $2^{14}$. However, $56\%$ of the distances is less than 15000 events. This means $56\%$ of the state control key is updated within 15000 events. The longest distance is 80000 events. All keys are updated below this maximum value.
 
  \begin{minipage}{0.5\textwidth}
 	\begin{figure}[H]
 		\centering
 		\includegraphics[width=0.8\columnwidth]{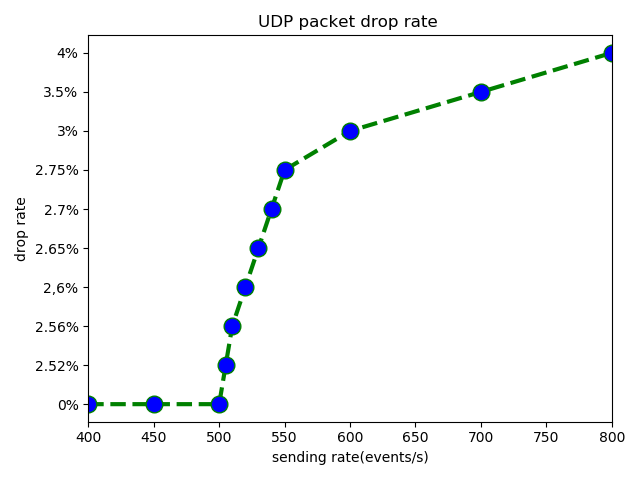}
 		\caption{\small Packet drop rate} \label{drop}
 	\end{figure}
 \end{minipage}
 \hfill
 \begin{minipage}{0.5\textwidth}
 	\begin{figure}[H]
 		\setlength{\belowcaptionskip}{-5mm}
 		\small
 		\centering
 		\includegraphics[width=0.8\textwidth]{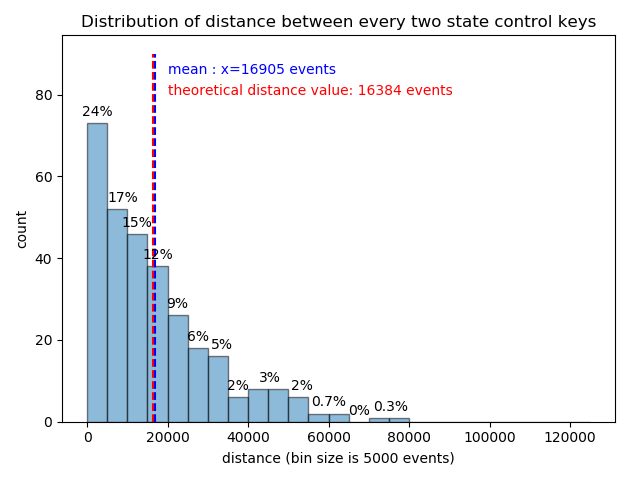}
 		\caption{Distribution of distance between every two state-controlled keys}\label{distribution}
 	\end{figure}
 \end{minipage}

\remove{
\section{Crash attack in distributed setting}\label{Distributed}
Considering crash attack in the distributed setting is one of the possible future directions. One example of distributed logging is state-machine replication systems in which the copy of log is stored on multiple systems, e.g. Raft \cite{Raft}, Paxos \cite{Paxos}, and Apache Kafka \cite{Kafka}. As we know in replicated systems we have copy of data in other systems, so even if cache is being removed data will not be lost. However, the state of the system can be lost, if attacker removes the index of last committed log and then truncates the log file. After recovery, it is not possible to detect the attack if the logs are stored plainly. To bootstrap this system, all removed logs should be sent from the leader which causes delay and unavailability on the system. So, the crash attack in distributed setting leads to delay attack in replicated systems. Finding a prevention mechanism which is aligned to requirements of replicated systems, like the system model (append-only log file), efficiency and fine-grained log access/verification can be a future research topic.
}

\end{document}